\documentclass[12pt]{article}
\usepackage[bookmarks=true, colorlinks=true, linkcolor=blue, citecolor=blue, breaklinks=true]{hyperref}
\usepackage{amsfonts}
\usepackage{rotating}
\usepackage{booktabs}
\usepackage{graphicx, graphics}
\usepackage{caption}
\usepackage{amsmath, amsthm, amssymb, natbib}
\usepackage{multirow}
\usepackage{tabularx}
\usepackage{ellipsis}

\usepackage[usenames]{color}
\usepackage[normalem]{ulem}
\definecolor{mygreen}{rgb}{0,.5,0}

\newtheorem{theorem}{Theorem}[section]
\newtheorem{proposition}{Proposition}[section]
\newtheorem{lemma}{Lemma}[section]
\newtheorem{corollary}{Corollary}[section]

\newtheorem{example}{Example}[section]

\newcommand{\cF}{{\cal F}}

\newcommand{\cA}{{\cal A}}

\newcommand{\cX}{{\cal X}}
\newcommand{\prob}{\mathbb{P}}

\newcommand{\expect}{\mathbb{E}}

\newcommand{\var}{V@R}

\newcommand{\hl}{\em \color{red}}

\begin{document}

\title{Surplus-Invariant, Law-Invariant, and Conic Acceptance Sets Must be the Sets Induced by Value-at-Risk\thanks{The paper was previously entitled ``Surplus-Invariant, Law-Invariant, and Positively Homogeneous Acceptance Sets Must be Induced by Value-at-Risk." The authors are grateful to the editor and the anonymous referees for their insightful comments that help to improve the paper. They are also thankful to Steven Kou for his constructive comments on the paper and for his host of the authors' visit to the Risk Management Institute at the National University of Singapore where the paper started. Xue Dong He is supported by the General Research Fund of the Research Grants Council of Hong Kong SAR (Project No. 14225916).}}

\author{Xue Dong He\thanks{Department of Systems Engineering and Engineering Management, The Chinese University of Hong Kong, Hong Kong, Email: xdhe@se.cuhk.edu.hk.} \and Xianhua Peng\thanks{HSBC Business School, Peking University, Nanshan District, Shenzhen, Guangdong, China, Email: xianhuapeng@phbs.pku.edu.cn.}}
\date{January 22, 2018}

\maketitle

\begin{abstract}
The regulator is interested in proposing a capital adequacy test by specifying an acceptance set for firms' capital positions at the end of a given period. This set needs to be surplus-invariant, i.e., not to depend on the surplus of firms' shareholders, because the test means to protect firms' liability holders. We prove that any surplus-invariant, law-invariant, and conic acceptance set must be the set of capital positions whose value-at-risk at a given level is less than zero. The result still holds if we replace conicity with num\'eraire-invariance, a property stipulating that whether a firm passes the test should not depend on the currency used to denominate its assets.\\





\end{abstract}

\baselineskip 18pt


\section{Introduction}

Suppose a firm has initial capital $c$, debt $d$, and thus asset value $c+d$ at time 0. Suppose the firm invests its assets in a portfolio that generates net random return $R$ in the period $[0, T]$. Suppose the firm needs to pay the interest $rd$ together with the face value $d$ to the creditor at time $T$, where $r$ is the interest rate. Then, the profit \& loss (P\&L) of the firm at the end of the period is $Y=(c+d)R-rd$. The Basel II Accord \citep{Basel06, Basel09} proposes a capital adequacy test based on the P\&L: the value-at-risk (VaR) of the P\&L at some confidence level, e.g., at 99\%, is calculated and then the firm is required to hold as much capital as the amount of the VaR. Thus, if we denote $\var_\alpha(Y)$ as the $\alpha$-level VaR of $Y$, then 
the firm's position is acceptable if and only if $\var_\alpha(Y)\le c$, i.e., if and only if $\var_\alpha(X)\le 0$, where $X:=Y+c= (c+d)(R+1)-(1+r)d$ is the capital position of the firm, i.e., the firm's assets net of its liability, at the end of the period.

As in the Basel Accords, the regulator is interested in proposing a capital adequacy test by specifying an acceptance set for firms' capital positions at the end of a given period. If a firm does not pass the test, it can 1) adjust its current portfolio, 2) inject new capital and hold it as cash, and 3) inject new capital and invest it in a portfolio.
Then, the regulator can check if the updated capital position is acceptable by performing the capital adequacy test again.

The positive part of a firm's capital position is the surplus for the equity holders and the negative part is defined as the {\em option to default} of the firm, which represents the portion of liabilities that is not paid off using the firm's assets. As argued by \citet{Staum2013:ExcessInvariance}, \citet{ContDeguestHe2013:LossBasedRiskMeasures}, \citet{KochMedinaEtal2015:CapitalAdequacyTests}, and \citet{KochMedinaEtal2016:Diversification}, the main purpose of a capital adequacy test is to protect liability holders; hence, whether a 
firm's capital position lies in the acceptance set 
should depend on only the negative part of the firm's capital position but not the positive part, a property referred to as surplus-invariance. 

In the present paper, we prove that surplus-invariant, law-invariant, and conic acceptance sets must be the sets induced by VaR, i.e., must be the set of capital positions whose VaR at a given level is less than zero. We assume neither convexity nor coherence, so the family of acceptance sets under our investigation include many that are not convex, such as those induced by the natural risk statistics proposed by \citet{KouEtal2012:ExternalRiskMeasures} and the distortion risk measures proposed by \citet{KouPeng2016:OnTheMeasurement}. Therefore, our result is different from those in \citet{KochMedinaEtal2015:CapitalAdequacyTests}  and \citet{KochMedinaEtal2016:Diversification}.

Law-invariance is a commonly adopted property that allows the regulator to perform statistical tests on whether a firm uses a valid model in the capital adequacy test. Conicity, which means that scaling a capital position by any positive constant does not affect the acceptability of the position, is assumed in the notion of coherent risk measures \citep{ArtznerPDelbaenFEberJMHeathD:99crm}. 
Conicity
is implied by num\'eraire-invariance, a property that is introduced by \citet{ArtznerEtal2009:RiskMeasures} and further investigated by
\citet{KochMedinaEtal2016:Diversification}. We also prove that surplus-invariant, law-invariant, and num\'eraire-invariant acceptance sets must be the sets induced by VaR.

The remainder of the paper is as follows: Section \ref{se:Main} delivers the main results and Section \ref{se:Conclusions} concludes. All proofs are placed in the Appendix.

\section{Main Results}\label{se:Main}
\subsection{Acceptance Sets}
Consider a probability space $(\Omega,\cF,\prob)$ and denote $\mathcal{L}^{0}(\Omega,\cF,\prob)$ as the set of all proper random variables on this space. Let $\mathcal{X}$ be a subset of $\mathcal{L}^{0}(\Omega,\cF,\prob)$ that contains $\mathcal{L}^{\infty}(\Omega,\cF,\prob)$, the set of bounded random variables. Here, $\mathcal{X}$ can include unbounded random variables; for example,  $\mathcal{X}$ can be $\mathcal{L}^{\infty}(\Omega,\cF,\prob)$, $\mathcal{L}^{0}(\Omega,\cF,\prob)$, and $\mathcal{L}^{p}(\Omega,\cF,\prob)$ for some $p\in[1,+\infty)$ that represents the set of random variables with finite $L^p$-norm.
Each element $X$ in $\mathcal{X}$ represents the {\em capital position} of a firm, i.e., the firm's assets net of its liability, at the end of a given period. Thus, the positive part of $X$, denoted as $X^+:=\max(X,0)$, is the {\em surplus} of the firm's shareholders and the negative part of $X$, denoted as $X^-:=-\min(X,0)$, is the {\em option to default} of the firm. The shareholders of the firm take the surplus but do not pay the option to default due to limited liability, so the liability holders cannot take the surplus but have to pay the option to default.

The regulator proposes a capital adequacy test by specifying an {\em acceptance set} $\cA$. 
We make the standard assumption in the risk measure literature (see, e.g., \citet{ArtznerPDelbaenFEberJMHeathD:99crm}, \citet{Staum2013:ExcessInvariance}, and \cite{KochMedinaEtal2015:CapitalAdequacyTests}) that whether a firm passes the capital adequacy test only depends on the capital position of the firm. Hence, $\cA$ is a subset of $\mathcal{X}$: a firm passes the test if and only if its capital position lies in $\cA$. We introduce the following properties for the acceptance set $\cA$:
\begin{enumerate}
	\setlength{\itemindent}{1.5ex}
	\item[(i)] {\em Surplus-invariance}: for any $X,Y\in\mathcal{X}$, if $X\in \cA$ and $X^-\ge Y^-$ almost surely (a.s.), then $Y\in \cA$.
	\item[(ii)]{\em Law-invariance}: for any $X,Y\in\mathcal{X}$, if $X\in \cA$ and $Y$ has the same distribution as $X$, then $Y\in \cA$.
	\item[(iii)]{\em Conicity}: for any $X\in\mathcal{X}$ and $\lambda>0$ such that $\lambda X\in \mathcal{X}$, if $X\in \cA$, then $\lambda X\in \cA$.
	\item[(iv)]{\em Num\'eraire-invariance}: for any $X\in\mathcal{X}$ and any strictly positive random variable $Z$ on $(\Omega, \cF, \prob)$ such that $ZX\in \mathcal{X}$, if $X\in \cA$, then $ZX\in \cA$.
	\item[(v)]
	{\em Truncation-closedness}: for any $X\in\mathcal{X}$, if $\min(\max(-d, X),d)\in \cA$ for any $d>0$, then $X\in \cA$.
\end{enumerate}


The surplus-invariance property is proposed by \citet{KochMedinaEtal2015:CapitalAdequacyTests},\footnote{ \citet{KochMedinaEtal2015:CapitalAdequacyTests} use a slightly different, but equivalent, definition of surplus-invariance; see Eq. (1.1) in their paper and Proposition 1 in \citet{KochMedinaEtal2016:Diversification}.
} extending the excess-invariance 
property of the shortfall risk measures and the loss-dependence property of the loss-based risk measures proposed by \citet{Staum2013:ExcessInvariance} and \citet{ContDeguestHe2013:LossBasedRiskMeasures}, respectively.\footnote{A risk measure $\rho$ is excess invariant if $\rho(X)=\rho(Y)$ for all $X$ and $Y$ such that $X^-=Y^-$ and is loss dependent if $\rho(X)=\rho(\min(X, 0))$ for any $X$. 
}
The surplus-invariance property stipulates that if firm A passes the test, then firm B whose option to default is smaller than that of firm A should also pass the test. 
This property is satisfied by the acceptance set that is associated with VaR, i.e., $\mathcal{A}=\{X\in \cX\mid \var_{\alpha}(X)\leq 0\}$, and the one that is associated with a shortfall risk measure, i.e., $\mathcal{A}=\{X\in \cX\mid \expect[l(X^-)]\leq c\}$, where $l$ is a nonconstant and increasing function. 
See \cite{Staum2013:ExcessInvariance}  and \citet{KochMedinaEtal2016:Diversification} for more examples and discussion.

If an acceptance set $\cA$ satisfies the surplus-invariance property, then $\cA$ satisfies the property that for any $X, Y\in \mathcal{X}$, if $X\le Y$ a.s. and $X\in \cA$, then $Y\in \cA$. In fact, the latter property is used in the definition of acceptance sets in \citet{KochMedinaEtal2015:CapitalAdequacyTests} and \citet{KochMedinaEtal2016:Diversification}.


Law-invariance is important for the regulator to use historical data to backtest the models used by firms in conducting the capital adequacy test.
Conicity stipulates that scaling the capital position of a firm by a positive constant does not change the acceptability of the firm. This property is assumed in coherent risk measures \citep{ArtznerPDelbaenFEberJMHeathD:99crm}, but is not required in convex risk measures \citep{FollmerHSchiedA:02convexriskmeasure,FrittelliMEmanuelaRG:02riskmeasure}. \cite{Weber2013} show that coherent risk measures applied to systems that exhibit price impact may induce convex risk measures that are not conic.

Num\'eraire-invariance
is introduced
in  \citet{ArtznerEtal2009:RiskMeasures} and further investigated in  \cite{KochMedinaEtal2016:Diversification}.
It means that
``the acceptance set should not
depend on the choice of the eligible (num\'eraire) asset" \citep[][p. 114]{ArtznerEtal2009:RiskMeasures}.
Num\'{e}raire-invariance can accommodate the situation that the assets and liabilities of firms may be denominated in different currencies.
\citet[][Section 6]{KochMedinaMunari2016:UnexpectedShortfalls} show that the acceptance set associated with 
VaR is num\'eraire-invariant but that associated with 
expected shortfall is not.
Clearly, num\'eraire-invariance implies conicity.\footnote{\citet[Proposition 5]{KochMedinaEtal2016:Diversification} show that a {\em closed} acceptance set $\cA$ is num\'eraire-invariant and monotone (i.e., for any $X$ and $Y$, if $X \leq Y$ a.s. and $X\in \cA$, then $Y \in \cA$) if and only if it is surplus-invariant and conic.} 
Without assuming num\'eraire-invariance, the currency risk can also be explicitly incorporated in capital adequacy tests by using vector-valued risk measures; see e.g., \cite{Jouini2004}.


The truncation-closedness property simply means that if any truncated version of the (possibly unbounded) capital position $X$ is acceptable, then $X$ itself is also acceptable. This property is similar to
the continuity axiom for the distortion risk measure $\rho$, which postulates that $\rho$ satisfies $\lim_{d\to +\infty}\rho(\min(d,X))=\lim_{d\to +\infty}\rho(\max(-d,X))=\rho(X)$ \citep{Wang1997:InsuranceMeasures}. The  truncation-closedness property automatically holds for any acceptance set if $\mathcal{X}=\mathcal{L}^{\infty}(\Omega, \cF, \prob)$ and is satisfied by some commonly used acceptance sets. For example,
for any distortion risk measure $\rho_h(X):=\int_0^{\infty}h(\prob(X>x))dx+\int_{-\infty}^0 (h(\prob(X>x))-1)dx$, where $h$ is a distortion function ($h$ is increasing with $h(0)=0$ and $h(1)=1$), $\mathcal{A}_{\rho_h}=\{X\in \cX\mid \rho_h(X)\leq 0\}$ is truncation-closed because $\lim_{d\to\infty}\rho_h(\min(\max(-d, X), d))=\rho_h(X)$.

\subsection{Value-at-Risk}

The current practice of capital adequacy tests in the Basel II Accord and the Solvency II is to use VaR; the Basel III Accord imposes that starting from January 1, 2018,  the expected shortfall at 97.5\% level under stressed scenarios will be used in the test of capital adequacy for market risk \citep{Basel2016}. To define VaR formally, we introduce some notations. For a given random variable $X$, denote $F_X$ as its (right-continuous) cumulative distribution function. Denote $F_X^{-1}$ as the left-continuous quantile function of $X$, i.e.,
$$F_X^{-1}(t) :=\sup\{x\in \mathbb R\mid F_X(x)<t\}=\inf\{x\in \mathbb R\mid F_X(x)\geq t\},\; t\in[0,1].$$
In particular, $F_X^{-1}(0)=-\infty$ and $F_X^{-1}(1)=\text{esssup}\, X:=\inf\{x\mid \prob(X\leq x)=1\}$ by definition. Furthermore,
$F_X^{-1}(t+)=\inf\{x\in \mathbb R\mid F_X(x)> t\}$ is the right-continuous quantile function of $X$  \citep[][Lemma A.19]{FollmerHSchied-4th}; in particular, $F_{X}^{-1}(0+)=\lim_{z\downarrow 0}F_X^{-1}(z)=\text{essinf}\, X:=\sup\{x\mid \prob(X\leq x)=0\}$, and $F_{X}^{-1}(1+)=+\infty$. It is well known that for any $x\in\mathbb{R}$ and $t\in[0,1]$, (i) $F_X(x)< t$ if and only if $F_X^{-1}(t)> x$ and (ii) $F_X(x-)\le t$ if and only if $F_X^{-1}(t+)\ge x$, where $F_X(x-):=\lim_{z\uparrow x}F_X(z)$.  
It is easy to show that\footnote{In fact, by Lemma A.27 in \citet{FollmerHSchied-4th}, $F_{\varphi(X)}^{-1}(z)=\varphi\big(F_X^{-1}(z)\big)$, for a.e. $z\in(0,1)$. Since both $F_{\varphi(X)}^{-1}(z)$ and $\varphi\big(F_X^{-1}(z)\big)$ are left-continuous, it follows that $F_{\varphi(X)}^{-1}(z)=\varphi\big(F_X^{-1}(z)\big)$, for any $z\in(0,1)$. Similar argument leads to $F_{\psi(X)}^{-1}(z)=\psi\big(F_X^{-1}((1-z)+)\big), \forall z\in (0,1)$.\label{footnote: argument_transform_quantile}} for any continuous, increasing function $\varphi: \mathbb{R}\to \mathbb{R}$, $F_{\varphi(X)}^{-1}(z)=\varphi\big(F_X^{-1}(z)\big),z\in(0,1)$ and for any continuous, decreasing function $\psi: \mathbb{R}\to \mathbb{R}$, $F_{\psi(X)}^{-1}(z)=\psi\big(F_X^{-1}((1-z)+)\big), z\in (0,1)$.

VaR of $X$ at level $\alpha$ is defined as the quantile function of $-X$ at level $\alpha$. Formally, for any $\alpha\in [0,1]$, we define the {\em lower $\alpha$-VaR} and {\em upper $\alpha$-VaR} of $X$ respectively to be
{\allowdisplaybreaks
	\begin{align}
	\var^{l}_\alpha(X):&=F^{-1}_{-X}(\alpha) = -F^{-1}_X((1-\alpha)+),\label{equ:var_l}\\
	\var^{u}_\alpha(X):&=F^{-1}_{-X}(\alpha+)=-F^{-1}_X(1-\alpha).\label{equ:var_u}
	\end{align}}%

For each $\alpha\in[0,1]$, we define the following three acceptance sets based on the default probability of the firm:
{\allowdisplaybreaks
	\begin{align}
	\cA^{-}_\alpha &:= \{X\in\mathcal{X} \mid \prob(X<0)<1-\alpha \},\label{equ:a_minus_alt}\\
	\cA^{0}_\alpha &:=
	\{X\in\mathcal{X} \mid \prob(X\le -\epsilon)<1- \alpha\text{ for any }\epsilon>0\},\label{equ:a_zero_alt}\\
	\cA^+_\alpha& := \{X\in\mathcal{X} \mid \prob(X<0)\le 1-\alpha \}.\label{equ:a_plus_alt}
	\end{align}}%

In the Basel II Accord, a capital position $X$ is acceptable if and only if its VaR at certain level $\alpha$ (e.g., 99\%) is less than or equal to zero. The following Proposition \ref{prop:VaRAcceptanceSets} shows that $\cA_{\alpha}^-$, $\cA_{\alpha}^0$, and $\cA_{\alpha}^+$ are indeed induced by VaR.
Moreover, these acceptance sets are surplus-invariant, law-invariant, conic, and truncation-closed; furthermore, $\cA_{\alpha}^-$ and $\cA_{\alpha}^+$ are num\'eraire-invariant.

%

\begin{proposition}\label{prop:VaRAcceptanceSets}
	(i) For any $\alpha\in [0,1]$, $\cA^-_\alpha$, $\cA^0_\alpha$, and $\cA^+_\alpha$ satisfy
	{\allowdisplaybreaks
		\begin{align}
		\cA^{-}_\alpha &=\{X\in \mathcal{X}\mid \var^u_{\alpha+\delta}(X)\le 0 \text{ for some }\delta\in (0, 1-\alpha)\},\label{equ:a_minus}\\
		\cA^{0}_\alpha &=\{X\in \mathcal{X}\mid \var^u_{\alpha}(X)\le 0\},\label{equ:a_zero}\\
		\cA^{+}_\alpha &=\{X\in \mathcal{X} \mid \var^l_\alpha(X)\le 0 \}.\label{equ:a_plus}
		\end{align}}%
	In particular, $\cA^-_1=\cA^0_1=\emptyset$ and $\cA_1^+=\{X\in\mathcal{X}\mid X\ge 0\ \text{a.s.}\}$.
	
	(ii) For any $\alpha\in[0,1]$, $\cA^-_\alpha$, $\cA^0_\alpha$, and $\cA^+_\alpha$ are surplus-invariant, law-invariant, conic, and truncation-closed acceptance sets, and
	$\cA^-_\alpha\subseteq \cA^0_\alpha\subseteq \cA^+_\alpha $; moreover, $\cA_\alpha^-$ and $\cA_\alpha^+$ are num\'eraire-invariant. For any $0\le \alpha_1<\alpha_2\le 1$, $\cA^-_{\alpha_1}\supseteq \cA^+_{\alpha_2}$.
\end{proposition}
%

Note that the surplus-invariance of $\cA_{\alpha}^+$ is already shown in
\citet[][Example 4.1]{Staum2013:ExcessInvariance} and 
\citet[][Example 1]{KochMedinaEtal2016:Diversification}. It is also worth noting that the acceptance set $\cA_1^+$ requires that a firm is acceptable if and only if it never defaults, so it is only formally a VaR-induced acceptance set. Finally, 
the differences between $\cA^-_\alpha$, $\cA^0_\alpha$, and $\cA^+_\alpha$ are minor, although in general $\cA^-_\alpha\subsetneq \cA^0_\alpha$ and $\cA^0_\alpha\subsetneq \cA^+_\alpha$. Moreover, $\cA^0_\alpha$ is not num\'eraire-invariant in general, as shown by Example \ref{ex:VaRNotNI} in Appendix \ref{appx:Examples}.

\subsection{Surplus-Invariant, Law-Invariant, Conic, and Truncation-Closed Acceptance Sets Must be the Sets Induced by VaR}


We prove that any surplus-invariant, law-invariant, conic (or num\'eraire-invariant), and truncation-closed acceptance set must be a set induced by VaR. To this end, we follow the literature of law-invariant risk measures (e.g., \citet{Kusuoka:01licrm}, \citet{JouiniESchachermayerWTouziN:06rm}, and \citet{FrittelliMEmanuelaRG:05lawinvariantriskmeasure}) to restrict ourselves to {\em atomless} probability spaces that support a uniform random variable. Our results also hold for {\em equi-probable} probability spaces ($(\Omega, \cF, \prob)$ is an equi-probable probability space if $\Omega=\{\omega_1, \ldots, \omega_n\}$ and $\prob(\omega_i)=1/n$, $\forall i$); see Appendix \ref{appx:EqualProbable}.

\begin{theorem}\label{th:Main}
	Assume $(\Omega, \cF,\prob)$ is atomless. Then,
	
	(i) If $\cA$ is a surplus-invariant, law-invariant, conic, and truncation-closed acceptance set, then there exists $\alpha\in[0,1]$ such that either $\cA^-_\alpha\subseteq \cA\subseteq \cA^0_\alpha$ or $\cA=\cA^+_\alpha$.
	
	(ii) $\cA$ is a surplus-invariant, law-invariant, num\'eraire-invariant, and truncation-closed acceptance set if and only if there exists $\alpha\in[0,1]$ such that either $\cA=\cA^-_\alpha$ or $\cA=\cA^+_\alpha$.
\end{theorem}

Theorem \ref{th:Main} shows that surplus-invariant, law-invariant, conic (or num\'eraire-invariant), and truncation-closed acceptance sets are essentially the sets induced by VaR. As a technical comment, there exists a surplus-invariant, law-invariant, conic, and truncation-closed acceptance set $\cA$ such that $\cA^-_\alpha\subsetneq \cA\subsetneq \cA^0_\alpha$; see Example \ref{ex:SINotVaR} in Appendix \ref{appx:Examples}.\footnote{This example also illustrates that for a surplus-invariant acceptance set, conicity is not equivalent to num\'eraire-invariance without any topological assumption on the set. For a comparison, see \citet[Proposition 5]{KochMedinaEtal2016:Diversification}.}

It is worth noting that an acceptance set $\cA$ that satisfies the properties in Theorem \ref{th:Main} cannot be induced by expected shortfall unless $\cA$ is trivial, i.e., is the empty set or the whole set of capital positions in $\mathcal{X}$ whose expected shortfall is well defined (i.e., finite). In fact, for the sake of contradiction, suppose such a nontrivial $\cA$ is induced by the expected shortfall at certain level $\beta \in [0,1]$; i.e., $\cA = \{X\in \mathcal{X} \mid \text{ES}_{\beta}(X)\in B\}$, where $B\subseteq \mathbb R$ and $\text{ES}_{\beta}(X):=\frac{1}{1-\beta}\int_{\beta}^1 \var^{l}_\alpha(X) d\alpha=\frac{1}{1-\beta}\int_{\beta}^1 F_{-X}^{-1}(\alpha)d\alpha$ for $\beta\in[0, 1)$ and $\text{ES}_{1}(X):=\text{essup}\,(-X)$. Note that $\text{ES}_{\beta}(X)$ is well defined for $X$ whose negative part is integrable when $\beta\in(0,1)$, for $X$ that is integrable when $\beta=0$, and for $X$ that is bounded from below when $\beta=1$. Now, for
any constant $c\ge 0$ and $X\in \mathcal{X}$, we have $c^-=0\le X^-$ . Recalling that $\cA$ is nonempty and surplus-invariant, we must have $c\in \cA$, implying that $(-\infty,0]\subseteq B$. On the other hand, if there exists $X\in \cA$ such that $\text{ES}_{\beta}(X)>0$, then it follows from the conicity of $\cA$ that $(0, \infty)\subseteq B$, which implies that $B=\mathbb{R}$ and thus $\cA = \{X\in \mathcal{X}\mid \text{ES}_{\beta}(X)\in \mathbb{R}\}$, a trivial set, which leads to a contradiction. Hence, there does not exist $X\in \cA$ such that $\text{ES}_{\beta}(X)>0$, implying $B=(-\infty,0]$; i.e., $\cA = \{X\in  \mathcal{X}\mid \text{ES}_{\beta}(X)\leq 0\}$. And then, it follows from \citet[][Example 5]{KochMedinaEtal2016:Diversification} that $ \{X\in  \mathcal{X}\mid \text{ES}_{\beta}(X)\leq 0\}$ is not surplus-invariant, which contradicts the assumption that $\cA$ is surplus-invariant. In addition, even
the trivial set $\cA = \{X\in \mathcal{X}\mid \text{ES}_{\beta}(X)\in \mathbb{R}\}$ may not satisfy the properties in Theorem \ref{th:Main}. For example, if $\mathcal{X}=\mathcal{L}^0(\Omega, \mathcal{F}, \prob)$, then $\cA = \{X\in \mathcal{X}\mid \text{ES}_{\beta}(X)\in \mathbb{R}\}$ is not truncation-closed.

Theorem \ref{th:Main} can be simplified if we further require that an acceptance set is closed with respect to some suitable topology on $\mathcal{X}$. In particular, the space $\mathcal{L}^{0}(\Omega,\cF,\prob)$ is a complete metric space under the Ky Fan metric (the metric is defined as $d(X,Y):=\expect[\min(|X-Y|, 1)]$, $\forall X, Y$), which metrizes convergence in probability; see, e.g., \citet[][Theorems 9.2.2 and 9.2.3, pp.289--290]{Dudley-2004}. Therefore, a  subset of $\mathcal{L}^{0}(\Omega,\cF,\prob)$ is closed under the Ky Fan metric topology if and only if it is closed under convergence in probability. Hence, we consider the following topological property of the acceptance set $\cA$:
\begin{enumerate}
	\setlength{\itemindent}{1.5ex}
	\item[(vi)]{\em Convergence-in-probability-closedness}: if $X_n\to X\in \mathcal{X}$ in probability as $n\to \infty$ and $X_n\in \cA$ for all $n$, then $X\in \cA$.
\end{enumerate}
Since $\min(\max(-d, X),d)\to X$ a.s. as $d\to \infty$, it follows that convergence-in-probability-closedness implies truncation-closedness. Theorem \ref{th:Main} can be simplified to be the following Theorem \ref{thm:closed} if we require that an acceptance set is closed with respect to the metric topology of $\mathcal{L}^{0}(\Omega,\cF,\prob)$, or, equivalently, closed with respect to convergence in probability.

\begin{theorem}\label{thm:closed}
	Assume $(\Omega, \cF,\prob)$ is atomless. Then, $\cA$ is a nonempty, surplus-invariant, law-invariant, conic (or num\'eraire-invariant), and convergence-in-probability-closed acceptance set if and only if there exists $\alpha\in[0,1]$ such that $\cA = \cA^+_\alpha$.
\end{theorem}

The proof of Theorem \ref{thm:closed} is obtained by applying Theorem \ref{th:Main} and verifying that for $\alpha\in[0,1)$, $\cA^+_{\alpha}$
is the closure of $\cA^-_{\alpha}$ in $\mathcal{X}$
under the metric topology of $\mathcal{L}^{0}(\Omega,\cF,\prob)$. Similar results can be obtained if one considers  $\mathcal{X}\subseteq\mathcal{L}^{p}(\Omega,\cF,\prob)$ for some $p\in[1,+\infty)$ and the closedness with respect to the $L^p$-norm, but not for the case of $\mathcal{X}=\mathcal{L}^{\infty}(\Omega,\cF,\prob)$ and the closedness with respect to the $L^\infty$-norm or the weak-star topology; see Appendix \ref{appx:Topology}.

\section{Conclusions}\label{se:Conclusions}
The regulator is interested in proposing a capital adequacy test by specifying an acceptance set for firms' capital positions at the end of a given period. Such test aims to protect firms' liability holders, so whether a firm passes the test should not depend on the surplus of the firm's shareholders, a property referred to as surplus-invariance. We prove that surplus-invariant, law-invariant, and conic acceptance sets must be the sets induced by VaR, i.e., must be the set of capital positions whose VaR at a given level is less than or equal to zero.
We also prove that surplus-invariant, law-invariant, and num\'eraire-invariant acceptance sets must be the sets induced by VaR.

Our results contribute to show the potential conflicts between competing properties of capital adequacy tests.
\citet{KochMedinaEtal2015:CapitalAdequacyTests} show that
the class of acceptance sets that are convex and num\'eraire-invariant (or
equivalently coherent and surplus-invariant) is limited to scenario tests; in addition, law-invariance further narrows the domain to the positive cone alone. We show that, if one drops convexity, the choice is enlarged only by VaR-induced acceptance sets. On the other hand, if one drops conicity or num\'eraire-invariance, then there is a rich class of surplus-invariant, law-invariant, and convex acceptance sets; see for instance \citet{ContDeguestHe2013:LossBasedRiskMeasures}, \citet{KochMedinaEtal2015:CapitalAdequacyTests}, and \citet{KochMedinaEtal2016:Diversification}. As a result, if one finds the range of choices of acceptance sets too restrictive, then one needs to sacrifice one of the properties
under investigation and is therefore forced to set up criteria to prioritize competing properties such as convexity, coherence, surplus-invariance, num\'eraire-invariance, and law-invariance. These properties might be desirable
if considered alone but lead to a narrow range of choices when  put together.\footnote{\label{fn:acknow_para}We are grateful to an anonymous referee for phrasing most of the sentences of this paragraph in his/her referee reports for our paper.}

%
%
%

\appendix

\section{Proofs}

\begin{proof}[Proof of Proposition \ref{prop:VaRAcceptanceSets}.]

(i)	For $\alpha=1$, it follows from the definitions in \eqref{equ:a_minus_alt}, \eqref{equ:a_zero_alt} and \eqref{equ:a_plus_alt} that $\cA^-_1=\cA^0_1=\emptyset$ and $\cA_1^+=\{X\in\mathcal{X}\mid X\ge 0\ \text{a.s.}\}$, which in combination with $\var^u_1(X)=-F_X^{-1}(0)=+\infty$ and $\var^l_1(X)=-F_X^{-1}(0+) = -\text{essinf}\, X$, imply that \eqref{equ:a_minus}, \eqref{equ:a_zero}, and \eqref{equ:a_plus} hold for $\alpha=1$.

We then prove \eqref{equ:a_minus}, \eqref{equ:a_zero}, and \eqref{equ:a_plus} for $\alpha\in [0, 1)$. Recall that for any random variable $X$, $x\in \mathbb{R}$, and $t\in [0,1]$, $F_X(x)< t\Leftrightarrow F_X^{-1}(t)> x$ and $F_X(x-)\le t\Leftrightarrow F_X^{-1}(t+)\ge x$. We will show that 
$F_X^{-1}(1-\alpha-\delta)\ge 0\ \text{for some}\ \delta\in(0,1-\alpha) \Leftrightarrow F_X^{-1}((1-\alpha-\delta')+)\ge 0\ \text{for some}\ \delta'\in(0,1-\alpha)$. The ``$\Rightarrow$" direction follows from $F_X^{-1}((1-\alpha-\delta)+)\geq F_X^{-1}(1-\alpha-\delta)$. To show the ``$\Leftarrow$" direction, suppose $F_X^{-1}((1-\alpha-\delta')+)\ge 0\ \text{for some}\ \delta'\in(0,1-\alpha)$. Let $\delta=\delta'/2$. Then, the monotonicity of $F_X^{-1}(\cdot)$ implies that $F_X^{-1}(1-\alpha-\delta)\geq F_X^{-1}((1-\alpha-\delta')+)\ge 0$.

Then, by \eqref{equ:var_u}, $\var^u_{\alpha+\delta}(X)\le 0\ \text{for some}\ \delta \in(0,1-\alpha)\Leftrightarrow F_X^{-1}(1-(\alpha+\delta))\ge 0\ \text{for some}\ \delta \in(0,1-\alpha)\Leftrightarrow F_X^{-1}((1-(\alpha+\delta'))+)\ge 0\ \text{for some}\ \delta' \in(0,1-\alpha)\Leftrightarrow 1-(\alpha+\delta')\ge F_X(0-)\ \text{for some}\ \delta' \in(0,1-\alpha)\Leftrightarrow \prob(X< 0) = F_X(0-)<1-\alpha$. Thus, \eqref{equ:a_minus} holds. Similarly, by \eqref{equ:var_u}, $\var^u_{\alpha}(X)\le 0 \Leftrightarrow F_X^{-1}(1-\alpha)\ge 0 \Leftrightarrow -\epsilon<F_X^{-1}(1-\alpha)\ \text{for any}\ \epsilon>0\Leftrightarrow \prob(X\le -\epsilon)=F_X(-\epsilon)<1-\alpha\ \text{for any}\ \epsilon>0$. Thus, \eqref{equ:a_zero} holds. Finally, by \eqref{equ:var_l}, $\var^l_{\alpha}(X)\leq 0\Leftrightarrow F_{X}^{-1}((1-\alpha)+)\geq  0 \Leftrightarrow \prob(X<0)=F_X(0-)\leq  1-\alpha$. Thus, \eqref{equ:a_plus} holds.

(ii) By definition, $\cA^-_\alpha$, $\cA^0_\alpha$, and $\cA^+_\alpha$ are law-invariant and conic. Moreover, it follows from \eqref{equ:a_minus_alt} that, for any $X\in \cA^-_\alpha$ and any $Y\in \cX$ satisfying $Y^-\le X^-$ a.s., $\prob(Y < 0) = \prob(-Y^-<0)\leq \prob(-X^-<0) = \prob(X < 0)<1-\alpha$, which implies $Y\in \cA^-_\alpha$. Hence, $\cA^-_\alpha$ is surplus-invariant. A similar argument
with the application of \eqref{equ:a_zero_alt} and \eqref{equ:a_plus_alt}
yields that $\cA^0_\alpha$ and $\cA^+_\alpha$ are also surplus-invariant. For any $X\in\mathcal{X}$ and any $d>0$, it holds that $\prob(\min(\max(-d, X),d) < 0)=\prob(X<0)$, which in combination with \eqref{equ:a_minus_alt} implies that $\cA^-_\alpha$ is truncation-closed. A similar argument yields that $\cA^0_\alpha$ and $\cA^+_\alpha$ are also truncation-closed.

Note that for any $X\in \mathcal{X}$ and any strictly positive random variable $Z$ such that $ZX\in \mathcal{X}$, $\prob(ZX<0)=\prob(X<0)$. Then, it follows from \eqref{equ:a_minus_alt} and \eqref{equ:a_plus_alt} that  $\cA^-_\alpha$ and $\cA^+_\alpha$ are num\'eraire-invariant.

Finally, it follows from \eqref{equ:a_minus_alt}, \eqref{equ:a_zero_alt}, and \eqref{equ:a_plus_alt} that $\cA^-_\alpha\subseteq \cA^0_\alpha\subseteq \cA^+_\alpha$ for any $\alpha\in[0,1]$ and $\cA^-_{\alpha_1}\supseteq \cA^+_{\alpha_2}$ for any $0\le \alpha_1<\alpha_2\le 1$.
\end{proof}

\begin{proof}[Proof of Theorem \ref{th:Main}.]
The proof of part (i) is as follows. If $\cA=\emptyset$, then $\cA=\cA^-_1=\cA^0_1$, so we assume $\cA$ to be nonempty in the following. Define
\begin{equation}\label{equ:alpha_def}
\alpha:= 1-\sup_{X\in\cA}\prob(X< 0)=1-\sup_{X\in\cA}F_X(0-).
\end{equation}
If $\alpha=1$, then by definition of $\alpha$, $\cA \subseteq \{X\in\mathcal{X}\mid X\ge 0\ \text{a.s.}\}=\cA^+_1$. On the other hand, because $\cA$ is nonempty and surplus-invariant, $\cA^+_1\subseteq \cA$. Therefore, $\cA=\cA^+_1$.

Now, suppose $\alpha\in[0,1)$.  First, we will show $\cA^-_\alpha\subseteq \cA$; i.e., for any $Z\in \cA^-_{\alpha}\subseteq \mathcal{X}$, we will show that $Z\in \cA$. Since $\cA$ is truncation-closed, we only need to show that for any fixed $d>0$, $Y:=\min(\max(-d, Z),d)\in \cA$. Since $Z\in \cA^-_{\alpha}$, it follows from \eqref{equ:a_minus_alt} that $\prob(Y<0)=\prob(Z<0)<1-\alpha$. Hence, there exists $\delta \in (0,1-\alpha)$ such that $F_Y(0-)=\prob(Y<0)\le 1- \alpha-\delta$. By the definition of $\alpha$, there exists $X\in \cA$ such that $F_X(0-)=\prob(X<0)> 1-\alpha-\delta$. Because $(\Omega, \cF,\prob)$ is atomless, there exists a uniformly distributed random variable $U$ on $(\Omega, \cF,\prob)$ such that $X=F_X^{-1}(U+)$ a.s. (see, e.g., Lemma A.32 in \citet{FollmerHSchied-4th}). Define $\tilde Y:=F_Y^{-1}(U+)$. Then, it follows from Lemma A.23 in \citet{FollmerHSchied-4th} that $\tilde Y$ has the same distribution as $Y$. Because both $Y$ and $\tilde Y$ are bounded and thus contained in $\mathcal{X}$ and $\cA$ is law-invariant, we only need to show that $\tilde Y\in \cA$.

Because $F_X(0-)>1-\alpha-\delta$, we have $F_X^{-1}((1-\alpha-\delta)+)<0$. As a result, because $F_Y^{-1}(0+)\geq -d$, there exists $\epsilon>0$ such that $\epsilon F_Y^{-1}(0+)\ge F_X^{-1}((1-\alpha-\delta)+)$, which implies that
$\epsilon F_Y^{-1}(z+)\ge \epsilon F_Y^{-1}(0+)\ge F_X^{-1}((1-\alpha-\delta)+) \ge F_X^{-1}(z+)$ for any $z\in (0,1-\alpha-\delta]$. Moreover, because $F_Y(0-)\le 1-\alpha-\delta$, we have $F_Y^{-1}(z+)\ge F_Y^{-1}((1-\alpha-\delta)+)\ge 0,z\in(1-\alpha-\delta,1)$. Consequently, $\min\left (\epsilon F_Y^{-1}(z+),0 \right) \ge \min\left( F_X^{-1}(z+),0\right)$ for any $z\in (0,1)$ and thus $(\epsilon \tilde Y)^- = - \min\left (\epsilon F_Y^{-1}(U+),0 \right)\le - \min\left( F_X^{-1}(U+),0\right) =  X^-$ a.s. Because $\epsilon\tilde Y\in \mathcal{L}^\infty(\Omega, \cF,\prob)\subseteq\mathcal{X}$, $X\in \cA$, and $\cA$ is surplus-invariant, $\epsilon \tilde Y\in \cA$. Then, because $\tilde Y\in\mathcal{L}^\infty(\Omega, \cF,\prob)\subseteq \mathcal{X}$ and $\cA$ is conic, we obtain that $\tilde Y\in \cA$.

On the other hand, \eqref{equ:a_plus_alt} and the definition of $\alpha$ in \eqref{equ:alpha_def} imply $\cA\subseteq \cA^+_\alpha$. Thus, we have proved that $\cA^-_\alpha \subseteq \cA\subseteq \cA^+_\alpha$. Hence, there are only two cases: (1) $\cA^-_\alpha\subseteq \cA\subseteq \cA^0_\alpha$; (2) there exists $X\in \cA$ such that $X\notin \cA^0_\alpha$. To complete the proof, we only need to show that the second case leads to $\cA=\cA^+_\alpha$.

In fact, in the second case, let $X\in \cA$ be such that $X\notin \cA^0_\alpha$. Since $X\in \cA$, it follows from \eqref{equ:alpha_def} that $F_X(0-)\leq 1-\alpha$, which implies that $F_X^{-1}\big((1-\alpha)+\big)\ge 0$. Since $X\notin \cA^0_\alpha$, it follows from \eqref{equ:a_zero} that $\var^u_{\alpha}(X)>0$, which in combination with \eqref{equ:var_u} implies that $F_X^{-1}(1-\alpha)<0$.

Because $(\Omega,\cF,\prob)$ is atomless, there exists a uniform random variable $U$ such that $X=F_X^{-1}(U)$ a.s. (see, e.g., Lemma A.32 in \citet{FollmerHSchied-4th}). Now, for any $Z\in \cA^{+}_\alpha\subseteq \mathcal{X}$, we will show that $Z\in \cA$. Since $\cA$ is truncation-closed, we only need to show that for any fixed $d>0$, $Y:=\min(d,\max(-d, Z))\in \cA$. Since $Z\in \cA^{+}_\alpha$, $Y\in \mathcal{L}^\infty(\Omega,\cF,\prob)\subseteq \mathcal{X}$, and $\prob(Y<0)=\prob(Z<0)$, it follows from \eqref{equ:a_plus_alt} that $Y\in \cA^{+}_\alpha$ and $F_Y(0-)\leq 1-\alpha$, which implies that
$F_Y^{-1}((1-\alpha)+)\ge 0$. Define $\tilde Y:=F_Y^{-1}(U)$. Then, it follows from Lemma A.23 in \citet{FollmerHSchied-4th} that $\tilde Y$ has the same distribution as $Y$. Because $F_X^{-1}(1-\alpha)<0$ and $F_Y^{-1}(0+)\geq -d$, there exists $\epsilon>0$ such that $\epsilon F_Y^{-1}(0+)\ge F_X^{-1}(1-\alpha)$, which implies that $\epsilon F_Y^{-1}(z)\ge \epsilon F_Y^{-1}(0+) \ge F_X^{-1}(1-\alpha) \ge F_X^{-1}(z), \forall z\in(0,1-\alpha]$. Moreover, because $F_Y^{-1}(z)\ge F_Y^{-1}((1-\alpha)+)\ge 0,z\in(1-\alpha,1)$, we conclude that $\min\left (\epsilon F_Y^{-1}(z),0 \right) \ge \min\left( F_X^{-1}(z),0\right)$ for any $z\in (0,1)$. Consequently, $(\epsilon \tilde Y)^- = - \min\left (\epsilon F_Y^{-1}(U),0 \right)\le - \min\left( F_X^{-1}(U),0\right) =  X^-$ a.s. Since $\cA$ is surplus-invariant and $\epsilon \tilde Y\in \mathcal{L}^\infty(\Omega,\cF,\prob)\subset \mathcal{X}$, it follows that $\epsilon \tilde Y\in \cA$. Because both $Y$ and $\tilde Y$ are bounded and thus contained in $\mathcal{X}$, conicity of $\cA$ implies that $\tilde Y\in \cA$ and law-invariance of $\cA$ implies that $Y\in \cA$. Therefore, $\cA=\cA^+_\alpha$, which completes the proof of part (i).

The proof of part (ii) is as follows. The ``if" direction follows from Proposition \ref{prop:VaRAcceptanceSets}. To show the ``only if" direction, suppose $\cA$ is surplus-invariant, law-invariant,  num\'eraire-invariant, and truncation-closed. Since num\'eraire-invariance implies conicity, it follows from part (i) of the theorem that there exists $\alpha\in [0, 1]$, such that either $\cA^-_{\alpha}\subseteq \cA \subseteq \cA^0_{\alpha}$ or $\cA=\cA^+_{\alpha}$. There are two cases: $\alpha=1$ and $\alpha\in [0, 1)$. In the first case, it follows from \eqref{equ:a_minus_alt} and \eqref{equ:a_zero_alt} that $\cA^-_{\alpha}=\cA^0_{\alpha}=\emptyset$. Hence, either $\cA=\cA^-_{\alpha}=\emptyset$ or $\cA=\cA^+_{\alpha}$. In the second case, suppose for the sake of contradiction that $\cA^-_\alpha\subsetneq \cA \subseteq \cA^0_\alpha$. Then, there exists $X\in \cA\subseteq \cA^0_\alpha$ such that
$X\notin \cA^-_\alpha$. Fix $d>0$ and define $Y:=\min(d,\max(-d,X))$. Noting $Y\in \mathcal{L}^\infty(\Omega, \cF,\prob)\subseteq \mathcal{X}$ and $Y^-\leq X^-$, we conclude from the  surplus-invariance of $\cA$ that $Y\in \cA \subseteq \cA^0_\alpha$. It follows from
\eqref{equ:a_minus_alt} and $\prob(Y<0)=\prob(X<0)$ that $Y\notin \cA^-_\alpha$ and
$\prob(Y<0)\ge 1-\alpha$. Define $Z:=-\big(1/Y\big)\mathbf 1_{\{Y<0\}} + 	 \mathbf 1_{\{Y\ge 0\}}$. Then, $Z>0$ and $ZY = -\mathbf 1_{\{Y<0\}} + 	Y\mathbf 1_{\{Y\ge 0\}}\in \mathcal{L}^\infty(\Omega,\cF,\prob)\subseteq \mathcal{X}$. Moreover, for any $\epsilon\in(0,1)$, $\prob(ZY\le -\epsilon) = \prob(Y<0)\ge 1- \alpha$, which implies from \eqref{equ:a_zero_alt} that $ZY\notin \cA^0_\alpha$ and thus $ZY\notin \cA$. This contradicts the assumption that $\cA$ is num\'eraire-invariant.
\end{proof}

\begin{proof}[Proof of Theorem \ref{thm:closed}.]
To prove the ``if" part, by Proposition \ref{prop:VaRAcceptanceSets}, we only need to show that $\cA^+_{\alpha}$ is convergence-in-probability-closed. Suppose $X_k\in \cA^+_{\alpha}$, $k=1, 2, \ldots$ and $X_k\to X\in \mathcal{X}$ in probability as $k\to \infty$. Then, for any $\delta>0$, $\prob(X\leq -\delta)=\prob(X\leq -\delta, |X_k-X|\geq \delta/2)+\prob(X\leq -\delta, |X_k-X|<\delta/2)\leq \prob(|X_k-X|\geq \delta/2)+\prob(X_k\leq -\delta/2)\leq \prob(|X_k-X|\geq \delta/2)+1-\alpha$. Letting $k\to \infty$, we obtain $\prob(X\leq -\delta)\leq 1-\alpha$. Then, letting $\delta \downarrow 0$ yields $\prob(X < 0)\leq 1-\alpha$, which in combination with \eqref{equ:a_plus_alt} implies that $X\in \cA^+_{\alpha}$.

Next, we prove the ``only if" part. Recall part (i) of Theorem \ref{th:Main} and the convergence-in-probability-closedness of $\cA^+_{\alpha}$ in $\mathcal{X}$, $\alpha\in[0,1]$. Because $\cA$ is assumed to be nonempty and $\cA^-_1=\cA^0_1=\emptyset$, we only need to show that for any $\alpha\in[0,1)$, the convergence-in-probability-closure of $\cA^-_{\alpha}$ in $\mathcal{X}$ is $\cA^+_{\alpha}$, i.e., for any $X\in \cA^+_{\alpha}\setminus \cA^-_{\alpha}$, there exists a sequence $X_k$, $k=1, 2, \ldots$ such that $X_k\in \cA^-_{\alpha}$, $\forall k$ and $X_k\to X$ in probability as $k\to \infty$. For any $n=1,2,\dots$, denote $Y_n:=\min(n,\max(-n,X))\in \mathcal{L}^\infty(\Omega,\cF,\prob)\subseteq \mathcal{X}$. Because $\prob(Y_n<0)=\prob(X<0)$, we conclude from \eqref{equ:a_minus_alt} and \eqref{equ:a_plus_alt} that $Y_n \in\cA^+_{\alpha}\setminus \cA^-_{\alpha}$. In the following, we show that for each fixed $n$, we can find a sequence $Z^n_{k}$, $k=1, 2, \ldots$ in $\cA^-_{\alpha}$ such that it converges to $Y_n$ in probability as $k\to \infty$. Then, for any $k\in \mathbb{N}$,
there exists index $m_k$ such that $d(Z^k_{m_k}, Y_k)<1/k$, where $d(\cdot, \cdot)$ denotes the Ky Fan metric (i.e., $d(X,Y):=\expect[\min(|X-Y|, 1)]$, $\forall X, Y$). Define the sequence $X_k:=Z^k_{m_k}$, $k=1, 2, \ldots$. Then, $X_k\in \cA^-_{\alpha}$ and $X_k$ converges to $X$ in probability as $k\to \infty$.

For fixed $n$, because $Y_n \in\cA^+_{\alpha}\setminus \cA^-_{\alpha}$, we conclude from \eqref{equ:a_minus_alt} and \eqref{equ:a_plus_alt} that $\prob(Y_n<0) = 1-\alpha>0$. Denote $a:=\text{essinf}\, Y_n$; then $a\in [-n,0)$ because $\prob(Y_n<0)>0$ and $Y_n\ge -n$. There are two cases: (i) $\prob(Y_n=a)=0$ and (ii)  $\prob(Y_n=a)>0$. In case (i), define $Z^n_k:=Y_n\mathbf 1_{\{Y_n\geq a+\frac{1}{k}\}}\in \mathcal{L}^\infty(\Omega,\cF,\prob)\subseteq \mathcal{X}$, $\forall k$. By the definition of $a$, $\prob(Y_n< a+\frac{1}{k})>0$, $\forall k$. Hence, $\prob(Z^n_k < 0)=\prob(Y_n<0) - \prob(Y_n< a+\frac{1}{k})<\prob(Y_n<0)=1-\alpha$, $\forall k>-\frac{1}{a}$, which implies that $Z^n_k\in \cA^-_{\alpha}$, $\forall k>-\frac{1}{a}$. In addition, $\lim_{k\rightarrow +\infty}\prob(Y_n<a+1/k)=\prob(Y_n=a)=0$, implying that $Z^n_k\to Y_n$ a.s. and thus in probability as $k\to \infty$. In case (ii), by the intermediate value theorem for atomless measures (see, e.g., Proposition 215D on p.46 of \citet{Fremlin-2010}), there exists a sequence of decreasing subsets $\{Y_n=a\}\supset A_1 \supset A_2 \supset \cdots$, such that $\prob(A_k)=\frac{1}{k+1}\prob(Y_n=a)$, $\forall k$. Define $Z^n_k:=Y_n\mathbf 1_{\{\Omega\setminus A_k\}}\in \mathcal{L}^\infty(\Omega,\cF,\prob)\subseteq \mathcal{X}$. Then, $\prob(Z^n_k<0)=\prob(Y_n<0)-\prob(A_k)<1-\alpha$, which implies that $Z^n_k\in \cA^-_{\alpha}$, $\forall k$. In addition, for any $\epsilon>0$, $\prob(|Z^n_k-Y_n|\geq \epsilon)\leq \prob(A_k)\to 0$ as $k\to\infty$, which implies that $Z^n_k\to Y_n$ in probability as $k\to\infty$.
\end{proof}

\section{Topologies on $\mathcal{X}$}\label{appx:Topology}
In this section, we first show that Theorem \ref{thm:closed} still holds if we restrict $\mathcal{X}\subseteq\mathcal{L}^{p}(\Omega,\cF,\prob)$ for some $p\in[1,+\infty)$ and consider the closedness with respect to the $L^p$-norm.
\begin{corollary}\label{cor:closedLp}
	Assume $(\Omega, \cF,\prob)$ is atomless and $\mathcal{X}\subseteq\mathcal{L}^{p}(\Omega,\cF,\prob)$ for some $p\in[1,+\infty)$. Then, $\cA$ is a nonempty, surplus-invariant, law-invariant, conic (or num\'eraire-invariant), and $L^p$-closed (i.e., $\cA$ is closed in $\mathcal{X}$ under the $L^p$-norm) acceptance set if and only if there exists $\alpha\in[0,1]$ such that $\cA = \cA^+_\alpha$.
\end{corollary}
\begin{proof}
Following the same steps in the proof of Theorem \ref{thm:closed}, we only need to show that (i) for any $\alpha\in[0,1]$, $\cA^+_\alpha$ is closed in $\mathcal{X}$ under the $L^p$-norm and (ii) for any $\alpha\in[0,1)$, the closure of $\cA^-_\alpha$ in $\mathcal{X}$ under the $L^p$-norm is $\cA^+_\alpha$. Claim (i) is true because we already showed in the proof of Theorem \ref{thm:closed} that $\cA^+_\alpha$ is closed in $\mathcal{X}$ under the convergence in probability, and the topology induced by the $L^p$-norm is stronger than that induced by the convergence in probability. To show claim (ii), recall the sequence of random variables $Z^n_k$ in $\cA^-_\alpha$ for large enough $k$ that are constructed in the proof of Theorem \ref{thm:closed} to converge in probability to a given bounded random variable $Y_n$ in $\cA^+_\alpha$. Careful examination shows that $Z^n_k$ converges,  as $k\rightarrow +\infty$, to $Y_n$ in $L^p$-norm as well. Thus, claim (ii) is true.
\end{proof}

Theorem \ref{thm:closed} does not hold if we consider $\mathcal{X}=\mathcal{L}^{\infty}(\Omega,\cF,\prob)$ and the closedness with respect to the $L^\infty$-norm because the closure of $\cA^-_\alpha$ under this norm is a strict subset of $\cA^+_\alpha$ in general. Indeed, similar to the proof of Theorem \ref{thm:closed}, we can show that $\cA^+_\alpha$ is closed in $\mathcal{X}$ under the $L^\infty$-norm. Now, fix $\alpha\in [0,1)$ and consider $X\in \cA^+_\alpha$ such that $\prob(X=1)=\alpha$ and $\prob(X=-1)=1-\alpha$. For any $Y\in \cA^-_\alpha$, we have $\prob(Y<0)<1-\alpha$, so $\prob(Y\ge 0,X=-1)\ge \prob(X=-1)-\prob(Y<0)>0$. Consequently, the $L^\infty$-norm of $X-Y$ must be larger than or equal to 1, implying that $X$ is not in the closure of $\cA^-_\alpha$ under the $L^\infty$-norm.

Theorem \ref{thm:closed} does not hold either if we consider $\mathcal{X}=\mathcal{L}^{\infty}(\Omega,\cF,\prob)$ and the closedness with respect to the weak-star topology because $\cA^+_\alpha$ is not closed under this topology in general. Indeed, consider $\Omega=[0,1]$, $\cF$ to be the Borel $\sigma$-algebra, and $\prob$ to be the Lebesgue measure. We show that for any $\alpha\in(0,1)$, $\cA^+_\alpha$ is not closed under the weak-star topology. For any $X\in \mathcal{L}^{\infty}(\Omega,\cF,\prob)$, any $\epsilon>0$, and any finite subset $A$ of $\mathcal{L}^{1}(\Omega,\cF,\prob)$, denote $N(X,A,\epsilon):=\{Y\in \mathcal{L}^{\infty}(\Omega,\cF,\prob)\mid |\expect[YZ]-\expect[XZ]|<\epsilon, \forall Z\in A\}$. Then, $\{N(X,A,\epsilon)\mid X\in \mathcal{L}^{\infty}(\Omega,\cF,\prob), A\text{ is a finite subset of }\mathcal{L}^{1}(\Omega,\cF,\prob),\epsilon>0\}$ form a base of the weak-star topology; see the Definition V.3.2 in \citet{DunfordNSchwartzJ:88linearoperators}. Now, consider $X^*\equiv -(1-\alpha) \in \mathcal{X}\backslash \cA^+_\alpha$. In the following, we show that for any finite subset $A$ of $\mathcal{L}^{1}(\Omega,\cF,\prob)$ and any $\epsilon>0$, there exists $X'\in  N(X^*,A,\epsilon)\cap \cA^+_\alpha$; consequently, $\cA^+_\alpha$ is not closed under the weak-star topology.

For each $i\in\{1,\dots,n\}$, because $Z_i\in \mathcal{L}^{1}(\Omega,\cF,\prob)$ and the set of continuous functions is dense in $\mathcal{L}^{1}(\Omega,\cF,\prob)$, there exists a continuous function $\tilde Z_i(t),t\in[0,1]$ such that $\expect[|Z_i-\tilde Z_i|]<\epsilon/3$. Because $\tilde Z_i$'s are uniformly continuous, there exists an integer $m\ge 1$ such that $|\tilde Z_i(t)-\tilde Z_i(s)|<\epsilon/ (6\alpha(1-\alpha))$ for any $i\in\{1,\dots,n\}$ and any $s,t\in[0,1]$ such that $|s-t|\le 1/m$. Now, define
{\allowdisplaybreaks
\begin{align*}
X_m(t)=\begin{cases}
-1,& t\in\left[\frac{k-1}{m},\frac{k-1}{m}+\frac{1-\alpha}{m}\right ),\\
0, & t\in\left[\frac{k-1}{m}+\frac{1-\alpha}{m},\frac{k}{m}\right ),
\end{cases}
\quad k=1,\dots,m.
\end{align*}}%
Then, $\prob(X_m<0) = \prob(X_m=-1)=1-\alpha$, so $X_m\in \cA^+_\alpha$. We have
{\allowdisplaybreaks
\begin{align*}
&|\expect[\tilde Z_iX_m]-\expect[\tilde Z_iX^*]| = \left|\int_0^1 \tilde Z_i(t)X_m(t)dt - \int_0^1 \tilde Z_i(t)X^*(t)dt\right| \\
={}& \left|\sum_{k=1}^m\left[-\int_{\frac{k-1}{m}}^{\frac{k-1}{m}+\frac{1-\alpha}{m}} \tilde Z_i(t)dt  \right]+(1-\alpha)\int_0^1 \tilde Z_i(t)dt\right|\\
={}& \left|\sum_{k=1}^m\left[-\alpha \int_{\frac{k-1}{m}}^{\frac{k-1}{m}+\frac{1-\alpha}{m}} \tilde Z_i(t)dt  + (1-\alpha) \int_{\frac{k-1}{m}+\frac{1-\alpha}{m}}^{\frac{k}{m}} \tilde Z_i(t)dt\right]\right|\\
={}&\left|\sum_{k=1}^m\left[-\alpha \int_{\frac{k-1}{m}}^{\frac{k-1}{m}+\frac{1-\alpha}{m}} \left( \tilde Z_i(t)-\tilde Z_i\left(\frac{k-1}{m}+\frac{1-\alpha}{m}\right)\right)dt  \right.\right.\\
&\quad\quad\quad\left.\left.+ (1-\alpha) \int_{\frac{k-1}{m}+\frac{1-\alpha}{m}}^{\frac{k}{m}}\left( \tilde Z_i(t)-\tilde Z_i\left(\frac{k-1}{m}+\frac{1-\alpha}{m}\right)\right)dt\right]\right|\\
\le{}&  \sum_{k=1}^m\left[\alpha \int_{\frac{k-1}{m}}^{\frac{k-1}{m}+\frac{1-\alpha}{m}} \left|\tilde Z_i(t)-\tilde Z_i\left(\frac{k-1}{m}+\frac{1-\alpha}{m}\right)\right|dt \right.\\
&\quad\quad\quad\left. + (1-\alpha) \int_{\frac{k-1}{m}+\frac{1-\alpha}{m}}^{\frac{k}{m}}\left|\tilde Z_i(t)-\tilde Z_i\left(\frac{k-1}{m}+\frac{1-\alpha}{m}\right)\right|dt\right]\\
<{}& \epsilon/3,
\end{align*}}%
where the last inequality holds because $|\tilde Z_i(t)-\tilde Z_i(s)|<\epsilon/(6\alpha(1-\alpha))$ for any $s,t\in[0,1]$ such that $|s-t|<1/m$. As a result,
{\allowdisplaybreaks
\begin{align*}
|\expect[ Z_iX_m]-\expect[ Z_iX^*]|  \le & {} |\expect[\tilde Z_iX_m]-\expect[\tilde Z_iX^*]|  +  |\expect[\tilde Z_iX_m]-\expect[ Z_iX_m]|\\
&{} + |\expect[\tilde Z_iX^*]-\expect[ Z_iX^*]| \\
< &{} \epsilon/3 +  2\expect[|\tilde Z_i -Z_i|]  <\epsilon,
\end{align*}}%
where the second inequality holds because $|X_m|\le 1$ and $|X^*|\le 1$, and the third inequality is the case because $\expect[|\tilde Z_i -Z_i|] <\epsilon/3$.  Thus, $X_m\in  N(X^*,A,\epsilon)\subseteq N(X,A,\epsilon)$. On the other hand, $X_m\in \cA^+_\alpha$. Therefore, $\cA^+_\alpha$ is not closed under the weak-star topology.

The above example is due to Marcel Nutz in a private communication between him and the first author of the present paper.

\section{The Case of Equi-probable Probability Spaces}\label{appx:EqualProbable}
\begin{theorem}\label{coro:EquiProbable}
	Suppose $(\Omega,\cF,\prob)$ is an equi-probable probability space. Then, $\cA$ is a surplus-invariant, law-invariant, and conic (or num\'eraire-invariant) acceptance set if and only if it is the empty set or $\cA^+_\alpha$ for some $\alpha\in[0,1]$.
\end{theorem}
\begin{proof}[Proof of Theorem \ref{coro:EquiProbable}.]
Suppose $\Omega=\{\omega_1, \ldots, \omega_n\}$. Then, by definition, $\prob(\{\omega_i\})=1/n, i=1, \ldots, n$. First, we show that for any random variable $X$ on $(\Omega,\cF,\prob)$, there exist random variables $U$ and $\tilde U$ on $(\Omega,\cF,\prob)$ such that (i)  $X=F_X^{-1}(U+)=F_X^{-1}(\tilde U)$ and (ii) any random variable $Y$ on $(\Omega,\cF,\prob)$ has the same distribution as $F_{Y}^{-1}(U+)$ and $F_{Y}^{-1}(\tilde U)$. In fact, without loss of generality (by relabeling the states if necessary), suppose that $X(\omega_i)=x_k, n_{k-1}<i\leq n_k, k=1,\ldots, m$,
where $n_0=0<n_1<\cdots<n_m=n$, and $x_1 < x_2 < \cdots < x_m$. Then, we define two random variables $\tilde U$ and $U$ as follows:
$\tilde U(\omega_i) := i/n, U(\omega_i) := (i-1)/n, i=1, \ldots, n$.
Then, we have $F_X^{-1}(\tilde U(\omega_i))=F_X^{-1}(i/n)=x_k=X(\omega_i)$, for any $n_{k-1}<i\leq n_k$; hence, $F_X^{-1}(\tilde U)=X$. In addition,
$F_X^{-1}(U(\omega_i)+)=F_X^{-1}(((i-1)/n)+)=x_k=X(\omega_i)$, for any $n_{k-1}<i\leq n_k$; hence, $F_X^{-1}(U+)=X$. For any random variable $Y$ on $(\Omega,\cF,\prob)$, suppose that $\{Y(\omega)\mid \omega\in\Omega\}=\{y_1, \ldots, y_{\bar m}\}$, where $y_1 < \cdots < y_{\bar m}$, and $\prob(Y\leq y_k)=\bar n_k/n$, $k=1,\ldots, \bar m$, where $\bar n_0=0<\bar n_1 < \cdots < \bar n_{\bar m} = n$. Define a random variable $\bar Y$ as $\bar Y(\omega_i)=y_k$, for any $\bar n_{k-1}<i\leq \bar n_k$. Then, we have $F_Y^{-1}(\tilde U(\omega_i))=F_Y^{-1}(i/n)=y_k=\bar Y(\omega_i)$, for any $\bar n_{k-1}<i\leq \bar n_k$; hence, $F_Y^{-1}(\tilde U)=\bar Y$. In addition,
$F_Y^{-1}(U(\omega_i)+)=F_Y^{-1}(((i-1)/n)+)=y_k=\bar Y(\omega_i)$, for any $\bar n_{k-1}<i\leq \bar n_k$; hence, $F_Y^{-1}(U+)=\bar Y$. Since $\bar Y$ has the same distribution as $Y$, it follows that (ii) holds.

Because $\Omega$ is a finite set, all random variables on $(\Omega,\cF,\prob)$ are bounded and thus form the space of capital positions $\mathcal{X}$; in consequence, any acceptance set is truncation-closed. Then, using $U$ and $\tilde U$ as constructed above, {it can be easily varified that} the proof for part (i) of Theorem \ref{th:Main} can also go through for the case of equi-probable probability spaces; therefore, if $\cA$ is a surplus-invariant, law-invariant, and conic acceptance set, then there exists $\alpha\in[0,1]$ such that $\cA^-_\alpha\subseteq \cA\subseteq \cA^0_\alpha$ or $\cA=\cA^+_\alpha$. When such $\alpha=1$, we have $\cA^-_1=\cA^0_1=\emptyset$, so $\cA=\emptyset$ or $\cA = \cA^+_1$ in this case. When such $\alpha\in[0,1)$, we claim that $\cA^-_\alpha= \cA^0_\alpha = \cA^+_{\alpha'}$ for some $\alpha'\in(\alpha,1]$. Indeed, for any $X\in \cA_\alpha^0$, there are two cases: (i) $X\geq 0$. In this case, by definition of $\cA^-_\alpha$, $X\in \cA^-_\alpha$. (ii) there exists $\omega$ such that $X(\omega)<0$. Let $x^*:=\max\{X(\omega)\mid X(\omega) < 0, \omega\in \Omega\}$. Note that $\Omega$ is a finite set, so $x^*<0$. Then, it follows from
\eqref{equ:a_zero_alt} that $\prob(X \leq x^*) < 1-\alpha$, which implies that $\prob(X < 0) = \prob(X \leq x^*) < 1-\alpha$. By \eqref{equ:a_minus_alt}, we conclude $X\in \cA_\alpha^-$. Hence, $\cA_\alpha^-=\cA_\alpha^0$. Now, noting that the probability of any event $A$ on $(\Omega,\cF,\prob)$ is a multiple of $1/n$, $\prob(A)<1-\alpha$ if and only if $\prob(A)\le 1-\alpha'$, where $\alpha':=1-\big(\lceil n(1-\alpha)\rceil-1\big)/n$ and $\lceil x\rceil$ stands for the ceiling of $x$ (i.e., the smallest integer dominating $x$); moreover, because $\lceil n(1-\alpha)\rceil-1<n(1-\alpha)$ and $\alpha<1$, we conclude that $\alpha'\in (\alpha,1]$. In consequence, by \eqref{equ:a_minus_alt} and \eqref{equ:a_plus_alt}, $X\in \cA^-_\alpha$ if and only $X\in \cA^+_{\alpha'}$; i.e., $\cA^-_\alpha=\cA^+_{\alpha'}$.

Finally, combining the above discussion for the case $\alpha=1$ and for the case $\alpha\in[0,1)$, we conclude that $\cA$ is a surplus-invariant, law-invariant, and conic acceptance set if and only if it is the empty set or $\cA^+_\alpha$ for some $\alpha\in[0,1]$. The above argument still holds when conicity is replaced by num\'eraire-invariance.
\end{proof}

\section{Two Examples}\label{appx:Examples}
\begin{example}\label{ex:VaRNotNI}
	Suppose $(\Omega, \cF,\prob)$ is atomless and let $\mathcal{X}=\mathcal{L}^{\infty}(\Omega,\cF,\prob)$. Consider a uniformly distributed random variable $Y$ with support $(-1+\alpha,\alpha)$. It follows from \eqref{equ:a_minus_alt} and \eqref{equ:a_zero_alt} that $Y\in \cA^0_\alpha$ and $Y\notin \cA^-_\alpha$. Define $Z:=-\big(1/Y)\mathbf 1_{\{Y<0\}} + \mathbf 1_{\{Y\ge 0\}}$. Then, $Z>0$ and $ZY = -\mathbf 1_{\{Y<0\}} + Y\mathbf 1_{\{Y\ge 0\}}\in \mathcal{X}$. Moreover, for any $\epsilon\in (0,1)$, $\prob(ZY\le -\epsilon)=\prob(Y<0)=1-\alpha$, which in combination with \eqref{equ:a_zero_alt} implies that $ZY\notin \mathcal{A}^0_\alpha$. Therefore, $\mathcal{A}^0_\alpha$ is not num\'eraire-invariant.
\end{example}
\begin{example}\label{ex:SINotVaR}
	Suppose $(\Omega, \cF,\prob)$ is atomless and let $\mathcal{X}=\mathcal{L}^{\infty}(\Omega,\cF,\prob)$. Fix $\alpha\in(0,1)$ and consider $Y\in\mathcal{X}$ such that $Y$ has a uniform distribution on $(-1+\alpha, \alpha)$. Then,
	$F^{-1}_{Y}(z)=z-(1-\alpha)$, $z\in (0,1)$. Define $\mathcal{X}_1:=\big\{X\in\mathcal{X}\mid \exists\, \delta > 0\text{ such that }\linebreak\delta \min\big(F^{-1}_X(z),0\big)\ge \min\big(F^{-1}_Y(z),0\big), \forall\, z\in(0,1)\big\}$ and $\mathcal{X}_2:=\mathcal{X}\backslash \mathcal{X}_1$.
	One can see that $\mathcal{X}_i$'s are law-invariant and conic; consequently, so is $\cA:=\big(\mathcal{X}_1\cap \cA_{\alpha}^{0}\big) \cup\big(\mathcal{X}_2\cap\cA_{\alpha}^{-}\big)$.
	First, 
	for any random variable $X$, $\min(F_{X}^{-1}(z),0)=F_{-X^-}^{-1}(z),z\in(0,1)$ because $F_{\varphi(X)}^{-1}(z)=\varphi\big(F_X^{-1}(z)\big),z\in(0,1)$ for any continuous and increasing function $\varphi$ (see endnote \ref{footnote: argument_transform_quantile}).
	Second, it is obvious that $\cA^{-}_\alpha\subseteq \cA\subseteq \cA^{0}_\alpha$. Third, we will show that $\cA$ is surplus-invariant. Consider any $X_1\in \cA$ and $X_2\in \mathcal{X}$ such that $X_2^-\leq X_1^-$ a.s. There are two cases: (1) If $X_1\in \mathcal{X}_1$, then $X_1\in \cA_\alpha^0$. 
	Because $X_2^-\leq X_1^-$ a.s.
	and $\cA^0_\alpha$ is surplus-invariant, we conclude that $X_2\in \cA_\alpha^0$.
	Furthermore, $\min(F_{X_1}^{-1}(z),0)=F_{-X_1^-}^{-1}(z)\le F_{-X_2^-}^{-1}(z)=\min(F_{X_2}^{-1}(z),0),\forall z\in (0,1)$ because $X_2^-\le X_1^-$ a.s. Consequently, $X_2 \in \mathcal{X}_1$ and thus $X_2\in \cA$.
	(2) If $X_1\in \cX_2$, then $X_1\in \cA_\alpha^-$.
	Because  $X_2^-\le X_1^-$ a.s. and $\cA^-_\alpha$ is surplus-invariant, we conclude that $X_2\in \cA_\alpha^-\subseteq \cA$. Fourth, as $\cX=\mathcal{L}^{\infty}(\Omega,\cF,\prob)$, the acceptance set $\cA$ is truncation-closed.
	Finally, one can see that $Y\notin \cA_\alpha^{-}$ and $Y\in \mathcal{X}_1\cap \cA^0_\alpha\subseteq \cA$. Moreover, $\tilde Y$ with quantile function $F^{-1}_{\tilde Y}(z)=z-(1-\alpha)$, $z\ge 1-\alpha$ and $F^{-1}_{\tilde Y}(z)=-\sqrt{1-\alpha-z}$, $z<1-\alpha$ falls neither in $\mathcal{X}_1$ nor in $\cA^{-}_\alpha$, so $\tilde Y\notin \cA$; however, $\tilde Y\in \cA_{\alpha}^{0}$. To summarize, $\cA$ is a surplus-invariant, law-invariant, conic, and truncation-closed acceptance set and $\cA^{-}_\alpha\subsetneq \cA\subsetneq \cA^{0}_\alpha$.
\end{example}

\bibliographystyle{dcu}
\bibliography{ShortTitles,BibFile}

\end{document}